\documentclass[notitlepage,secnumarabic,twocolumn,graphics,tightenlines,superscriptaddress,nobibnotes,aps,prl]{revtex4-2}

\usepackage[english]{babel}

\usepackage[utf8]{inputenc}
\usepackage{tikz}
\usetikzlibrary{positioning}
\usepackage{multirow}

\usepackage{amsmath,amssymb,amsthm}
\newtheorem{theorem}{Theorem}
\newtheorem{proposition}[theorem]{Proposition}
\newtheorem{coro}[theorem]{Corollary}
\newtheorem{assumption}[theorem]{Assumption}

\newcommand{\bra}[1]{\langle #1 |}
\newcommand{\ket}[1]{| #1 \rangle}

\newcommand{\comm}[2]{\left[#1,#2\right]}

\newcommand{\CC}{\mathbb{C}}
\newcommand{\RR}{\mathbb{R}}

\newcommand{\su}{\mathfrak{su}}

\newcommand{\der}[3]{\frac{{\rm d}^{#1}#3}{{\rm d}{#2}^{#1}}}
\newcommand{\partder}[3]{\frac{\partial^{#1}#3}{\partial{#2}^{#1}}}

\newcommand{\abs}[1]{| #1 |}

\newcommand{\Sx}{\hat{S}_x}
\newcommand{\Sy}{\hat{S}_y}
\newcommand{\Sz}{\hat{S}_z}

\newcommand{\U}{\hat{U}}
\newcommand{\Ham}{\hat{H}}
\newcommand{ \T}{\hat{T}}

\newcommand{\tr}{\mathrm{tr}\,}

\newcommand{\bi}{\begin{itemize}\setlength\itemsep{0pt}}
\newcommand{\ei}{\end{itemize}}
\newcommand{\be}{\begin{enumerate}\setlength\itemsep{0pt}}
\newcommand{\ee}{\end{enumerate}}
\newcommand{\non}{\nonumber}
\newcommand{\bea}{\begin{eqnarray}}
\newcommand{\eea}{\end{eqnarray}}

\newcommand{\citearxiv}[1]{ \cite{#1}}
\newcommand{\citearxiveg}[1]{ \cite[e.g.]{#1}}
\newcommand{\arxiv}[1]{#1}
\newcommand{\nonarxiv}[1]{}

\begin{document}

\title{Quantum Optimal Control via Magnus Expansion and Non-Commutative Polynomial Optimization}
\author{Jakub Marecek}
\affiliation{IBM~Research}
\author{Jiri Vala}
\affiliation{IBM~Research}
\affiliation{Department of Theoretical Physics, Maynooth University, Ireland}
\affiliation{School of Theoretical Physics, Dublin Institute for Advanced Studies,  Ireland}
\date{\today}

\begin{abstract}
    Quantum optimal control has numerous important applications
    ranging from pulses in magnetic resonance imagining to laser control of chemical reactions and quantum computing. 
Our objective is to address two major challenges that have limited the success of applications of quantum optimal control so far: non-commutativity inherent in quantum systems and non-convexity of quantum optimal control problems involving more than three quantum levels. 
 Methodologically, we address the non-commutativity of the control Hamiltonian at different times by the use of Magnus expansion. 
 To tackle the non-convexity, we employ non-commutative polynomial optimisation and non-commutative geometry.
 As a result, we present the first globally convergent methods for quantum optimal control.
\end{abstract}

\maketitle

\section{Introduction}

Quantum optimal control is behind many recent advances in science and technology.
In Biology, quantum optimal control allows nuclear magnetic resonance (NMR) spectroscopy to study large biomolecules in solution \cite[e.g.]{skinner2003application,khaneja2005optimal}. 
In Chemistry, quantum optimal control in laser spectroscopy brings fundamental insights into reaction dynamics; laser control directs chemical reactions to a desired target or even enables a design of new chemical species and materials  \cite[e.g.]{assion1998control,rabitz2000whither,rice2000optical}. 
In Neurology and Neurosciences, quantum optimal control provides radio-frequency pulses yielding higher resolution \cite{xu2008designing} in  functional magnetic resonance imaging (fMRI), and hence better diagnoses with less time spent in the scanner.
In Photonics and Metrology, 
interferometers \cite{o2009photonic,eckle2008attosecond} utilise quantum optimal control as a means of designing semi-classical probes, which 
have applications from gravity-wave detection in cosmology \cite{adhikari2014gravitational} to phase-contrast microscopy and attosecond tunnelling in nanotechnology \cite{eckle2008attosecond}. In other areas of Physics, real-time quantum control improves  preparation of non-classical states of light \cite{sayrin2011real}, for example.
In quantum computing \cite{Nielsen_00}, better quantum optimal control provides faster and more accurate two-qubit gates \cite{motzoi2009simple,o2009photonic}, and multi-level operations in general,  
enabling fault-tolerant quantum computation  \cite{campbell2017roads} eventually. More generally, one can replace the application of an entire quantum circuit with a control signal.
Despite the importance, quantum optimal control remains little understood.


There are many excellent results\citearxiveg{Brockett1972,Jurdjevic1981,el1996subsemigroups,Boscain2014,boscain2015approximate,Borzi2017},  that establish conditions for controllability, as surveyed in \cite{d2007introduction,jurdjevic2016optimal},
but constructive, algorithmic approaches still leave space for improvement.
Only the optimal control of two-level closed quantum systems (e.g., pulse shaping for one-qubit gates) and three-level closed quantum systems (e.g., pulse shaping for one-qutrit gates) is essentially solved\citearxiv{d2001optimal,boscain2002optimal,garon2013time,romano2014geometric,albertini2016time,7981370}, because the problem is invex \cite{pechen2014coherent}. 
The control of systems involving more than three levels\citearxiv{albertini2003notions}, including the pulse-shaping for two-qubit gates, is essentially an open problem. 
First, most formulations seem to assume commutativity of the Hamiltonian at different times.
Second, the corresponding quantum optimal control on an $N$-level system is non-convex for $N \ge 4$ \cite{pechen2011there,Zhdanov2018,bondar2019uncomputability}, but only heuristics based on first-order optimality conditions are employed. 
There can hence be arbitrarily bad local optima \cite{pechen2011there,wu2011role,Zhdanov2018}, which the heuristics cannot escape. 

In this paper, we address both issues of non-commutativity and non-convexity of the problem by employing Magnus expansion \cite{Magnus_54,Blanes_09} and tools from non-commutative polynomial optimisation \cite{Pironio2010,burgdorf2016optimization}.
In addressing the non-commutativity side of the problem, our work improves most directly upon the work of Schutjens et al. \cite{Egger2013} and Theis et al. \cite{Theis2016}, who consider the lowest-order term of the Magnus expansion, also known as the average Hamiltonian, and derive conditions for all other terms being zero. In contrast to their approach (known as Weak Anharmonicity with Average Hamiltonian), we consider an arbitrary number of terms in the Magnus expansion. 
Our work complements research on Magnus expansion in numerical integration of the Schr{\"o}dinger equation  \cite[e.g.]{Blanes_09,singh2018high,kopylov2019magnus}, which so far has not been developed in the context of quantum control. 

In addressing the non-convexity of the problem, we utilise a hierarchy of progressively stronger convexifications. This improves upon all related work on quantum optimal control,
which guarantees only monotonic convergence to first-order critical points or local minima\citearxiv{peirce1988optimal,krotov1995global,reich2012monotonically} of a non-convex optimisation problem based on Pontryagin's maximum principle. 
This related work can be grouped into several clusters:
(i) 
derivative-free optimisation methods\citearxiveg{doria2011optimal,caneva2011chopped,Theis2016,bukov2018reinforcement}, including the so called chopped random basis \cite{doria2011optimal}\citearxiv{caneva2011chopped} method (CRAB),
(ii) gradient methods\citearxiv{peirce1988optimal}, including the so-called
 gradient ascent pulse engineering (GRAPE,  \cite{khaneja2005optimal,motzoi2009simple}\citearxiv{gambetta2011analytic,Egger2013}) and 
Krotov method \cite{tannor1992control}\citearxiv{krotov1995global,reich2012monotonically} 
as two prominent examples.
Faster convergence to local optima may be obtained using (iii) quasi-Newton\citearxiveg{de2011second,eitan2011optimal,machnes2011comparing,vinding2012fast},
and Newton-like methods such as \cite{ciaramella2015newton}.
Even the work on (vi) sequential convexifications \cite{kosut2013robust} is essentially heuristic, albeit addressing the challenge of non-convexity explicitly,
and being the closest to our method, in spirit.
We refer to \arxiv{\cite{mabuchi2005principles,d2007introduction,altafini2012modeling,jurdjevic2016optimal} for extensive surveys within control theory and to\citearxiv{peirce1988optimal,werschnik2007quantum,brif2010control,machnes2011comparing,glaser2015training} for tutorials and surveys aimed at the physics community. \cite{brockett1973lie,dirr2008lie} introduce the mathematical foundations well.}\nonarxiv{\cite{d2007introduction,jurdjevic2016optimal,glaser2015training} for surveys.}
In summary, 
our approach is the first to offer guarantees of asymptotic convergence to the global optimum of the quantum optimal control problem.


\section{The Problem}

Let us consider a finite $N$-dimensional quantum system whose time-evolution is governed by a Schr\"odinger equation. 
Given an initial condition $\U(0) = \hat{I}$,
where $\hat{I}$ is a unit matrix in $\CC^{N \times N}$,
a terminal time $T > 0$,
and a target unitary $\U^* \in U(N) \subset \CC^{N \times N}$, 
where 
$U(N)$ is the Lie group of $N \times N$ unitary operators or matrices,
we aim to control a time-dependent Hamiltonian $\Ham (t)$ over time $t \in [0, T]$.
That is, we seek a particular solution to the initial value problem for the Schr\"odinger equation\footnote{There are many formulations possible. In contrast to those involving unbounded operators, such as position and momentum, we consider only a finite dimension $N$; this is perfectly sufficient for many applications, including the control of a system of a finite number of quantum bits in quantum-computing applications.} 
\begin{align}
\label{schroedinger}
\partder{}{t}{} \U (t) = \hat{A}(t) \U (t)
\end{align}
where $\hat{A}(t) = \hat{H}(t)/i\hbar$ can explicitly be written in terms of controls $u_j(t): [0, T] \to \RR$ as
\begin{align}
\label{eq:hata}
    \hat{A}(t) = \sum_j u_j(t) ~\Ham_j/i\hbar. 
\end{align}

\begin{figure*}[htb]
\begin{tikzpicture}
\def\angle{90}%
\pgfmathsetlengthmacro{\xoff}{2cm*cos(\angle)}%
\pgfmathsetlengthmacro{\yoff}{0.55cm*sin(\angle)}%
\draw [thick, fill=gray!10] (\xoff,\yoff) circle[x radius=2.5cm, y radius=2cm] ++(2*\xoff,2*\yoff) node{$\textrm{ME}(R(\U(0)), 3)$};
\draw [thick, fill=gray!20] (0.5*\xoff,0.5*\yoff) circle[x radius=1.7cm, y radius=1.1cm] ++(1*\xoff,1*\yoff) node{$\textrm{ME}(R(\U(0)), 2)$};
\draw [thick, fill=gray!30] (0,0) circle[x radius=0.7cm, y radius=0.5cm] node{$R(\U(0))$};
\node [rectangle, right=1cm, text width=10cm] (eq1) at (2, 0) {
 \begin{minipage}{\textwidth}
\begin{align}
\Omega_1(T) &= \int_0^T dt_1~\hat{A}_1, \non \\
\Omega_2(T) &= \frac{1}{2} \int_0^T  dt_1  \int_0^{t_1}  dt_2 ~\comm{\hat{A}_1}{\hat{A}_2}, \non \\
\Omega_3(T) &= \frac{1}{6} \int_0^T  dt_1  \int_0^{t_1}  dt_2  \int_0^{t_2}  dt_3 \left(\comm{\hat{A}_1}{\comm{\hat{A}_2}{\hat{A_3}}} + \comm{\comm{\hat{A}_1}{\hat{A_2}}}{\hat{A}_3} \right), \non \\
\Omega_4(T) &= \frac{1}{12} \int_0^T  dt_1  \int_0^{t_1}  dt_2  \int_0^{t_2}  dt_3 \int_0^{t_3}  dt_4 \left( \comm{\comm{\comm{\hat{A}_1}{\hat{A}_2}}{\hat{A}_3}}{\hat{A}_4} \right. \non \\
& + \left. \comm{\hat{A}_1}{\comm{\comm{\hat{A}_2}{\hat{A}_3}}{\hat{A}_4}} + \comm{\hat{A}_1}{\comm{\hat{A}_2}{\comm{\hat{A}_3}{\hat{A}_4}}} + \comm{\hat{A}_2}{\comm{\hat{A}_3}{\comm{\hat{A}_4}{\hat{A}_1}}} \right), \non \\
\dots \non
\end{align}
\end{minipage}
};
\end{tikzpicture}
\caption{First terms of the Magnus expansion\arxiv{ (\ref{MagnusExpansion})}, utilising 
$\hat{A}_m = \hat{A}(t_m) = \hat{H}(t_m)/i\hbar$,
and a schematic illustration of (a two-dimensional projection of) the corresponding reachable sets (in possibly increasing dimensions). Notice that reachable set $\textrm{ME}(R(\U(0)), m)$ corresponds to $\Omega_m(T)$ and  $\textrm{ME}(R(\U(0)), m)$ is a subset of $\textrm{ME}(R(\U(0)), m+1)$.}
\label{fig:first}
\end{figure*}

In particular, we seek a solution that is optimal with respect to a given  functional $J$, while using controls $\{u_j(t)\}$ constrained to some set $\Upsilon$.
Formally, the quantum optimal control problem reads:
\begin{align}
    \displaystyle\min_{\U(t), \{u_j(t)\} \in \Upsilon} & J \left(\U (t), \{u_j(t)\} \right) \label{eq:QOC} \\
    \mbox{s.t.} ~~&\partder{}{t}{} \U (t) =  \left[\sum_j u_j(t) ~\Ham_j/i\hbar\right] \U (t), \non \\
    &\U(0) = \hat{I}.
    \non
\end{align}
where $J$ is the (objective) functional for the control problem, 
which is polynomially or semidefinite representable \cite{helton2007linear}, 
and $\Upsilon$ is a polynomially representable set.
Many examples of $J$, such as trace distance of $\U(T)$ from a target unitary $\U^*$, are discussed in the Supplementary Material\arxiv{ (p. \pageref{sec:functionals})}.

\arxiv{
\section{Notation and Definitions}

When the terminal time $T$, also known as horizon, is not fixed, the existence of $\{u_j(t)\}$ for any target $\U^*$  is variously known as complete controllability, operator controllability, or exact controllability, cf. \cite{albertini2003notions}.
In particular, a time-dependent Hamiltonian $\Ham (t)$ (system) is operator controllable if and only if 
for any $\U^* \in U(N)$ 
there exists a terminal time $T$ and an admissible control to drive the state $\U(t)$ in \eqref{schroedinger} from $\U(0) = \hat{I}$ to $\U(T) = \U^*$.

If a system is not operator controllable, we are concerned with the reachable set $R(\U(0))$ from a given initial state $\U(0)$.
In particular, $\U^* \in U(N)$ is an element of the reachable set $R(\U(0)) \subseteq U(N) \subset \CC^{N \times N}$ if and only if there exists a horizon $T$ and an admissible control to drive the state from $\U(0)$ to $\U(T) = \U^*$ in \eqref{schroedinger}.
With this definition, a system is operator controllable if and only if $R(\hat{I}) = U(N)$.


In the following we will focus on the special unitary group $SU(N)$ of $N \times N$ complex matrices with the unit determinant. The group $U(N)$ factorizes into a semidirect product of $SU(N)$ and the cyclic group $U(1)$ of complex numbers $z = e^{i \varphi}$ with the modulus $\abs{z} = 1$ which however have no observable consequences in quantum mechanics and can be ignored. 

The group $SU(N)$ is generated by elements of the Lie algebra $\su(N)$ in the sense that each element of the group is a complex exponential function of a Hermitian traceless matrix from the algebra (or a real exponential function of an anti-Hermitian traceless matrix as it is  more common in mathematics literature). An $\su(N)$ algebra element can be written as a real linear combination of suitably chosen $N^2 - 1$ basis elements. 

A Hamiltonian operator generates the unitary evolution operator in the same way as an element of a Lie algebra generates the corresponding element of a Lie group. To make the connection with the algebra $\su(N)$, we shift the Hamiltonian by a constant term to eliminate the $U(1)$ factor in the dynamics it generates. In the case $N = 2^n$, the Hamiltonian can be written in terms of the Pauli matrices and their tensor products.

The condition for the operator controllability in $SU(N)$ via quantum dynamics is related to what extent the time-dependent Hamiltonian $\hat{H}(t)$ explores the Lie algebra $\su(N)$. Notably, Borzi et al. \cite{Borzi2017}(Theorem 4.7) 
show that a necessary and sufficient condition for operator controllability in $SU(N)$ of the Equation \eqref{schroedinger} is that the Lie algebra generated by the Hamiltonian $\Ham (t)$  has dimension $N^2 - 1$.
This is an important existential condition but it provides no insight into a mechanism by which the corresponding Lie algebra is generated from the time-dependent Hamiltonian in the course of time evolution. For this, we have to go back to the initial value problem \eqref{schroedinger}. 
}

\section{Our Approach}

It is well known that the initial value problem \eqref{schroedinger} has a solution in the form of the Magnus expansion \cite{Magnus_54,Blanes_09}:

\nonarxiv{
\begin{align}
\label{MagnusExpansion}
\Omega(T) = \sum_{m = 1}^\infty \Omega_m (T), 
\end{align}
where the individual terms in the series require evaluations of increasingly more complex integrals involving nested commutators, as illustrated in Figure~\ref{fig:first}.
When the series $\Omega_m(T)$ is absolutely convergent, then $\U(t)$ can be written in the form
\begin{align}
\label{eq:solution}
    \U(t) = \exp \Omega(t). 
\end{align}
}

\arxiv{
\begin{theorem}[Magnus \cite{Magnus_54,Blanes_09}]
\label{Magnus_theorem}
Let $\hat{A}(t)$ be a known function of time $t$, and let $\U(t)$ be an unknown function satisfying \eqref{schroedinger} with $\U(0) = \hat{I}$. Let 
\begin{align}{\label{MagnusOmega}}
    \der{}{t}{\Omega} = \sum_{n = 0}^{\infty} \frac{B_n}{n!} ~\operatorname{ad}_\Omega^n~\hat{A} 
\end{align}
where $B_n$ are the Bernoulli numbers and $\operatorname{ad}_\Omega^n\hat{A}$ is a linear operator constructed recursively via $\operatorname{ad}_{\hat{A}}^{j}\hat{B} = \comm{\hat{A}}{\operatorname{ad}_{\hat{A}}^{j-1}\hat{B}}$ where $\operatorname{ad}_{\hat{A}}\hat{B} = \comm{\hat{A}}{\hat{B}}$, and $\operatorname{ad}_{\hat{A}}^{0}\hat{B} = \hat{B}$.
Integration of (\ref{MagnusOmega}), by iteration, leads to an infinite series:
\begin{align}
\label{MagnusExpansion}
\Omega(T) = \sum_{m = 1}^\infty \Omega_m (T), 
\end{align}
whose first terms are given in Figure~\ref{fig:first}.
When the series $\Omega_m(T)$ is absolutely convergent, then $\U(t)$ can be written in the form
\begin{align}
\label{eq:solution}
    \U(t) = \exp \Omega(t). 
\end{align}
\end{theorem}
}

In a key insight of this paper, we show that the nested commutators between the Hamiltonian at different times are instrumental in extending the reachable set. This expansion is accomplished via two distinct mechanisms. First, the nested commutators generate new linearly independent elements of the Lie algebra $\su(N)$ and hence increase the dimension of the reachable set. The operator controllability is accomplished when this dimension reaches $N^2 - 1$. 

The second mechanism is related to the controls $\{u_j(t)\}$, which figure in coefficients of the Lie algebra elements that are generated by the $k$ terms of the Magnus expansion. The controls are, in general, arbitrary functions of time which allows for an arbitrary linear combination of the Lie algebra elements. Hence the two mechanisms together result in expanding the dimension of the reachable set as a Lie algebra and also in densely covering this set by control functions and their integrals.
We denote the expanded reachable set obtained with the lowest $m$ terms in the Magnus series by $\textrm{ME}(R(\U(0)), m)$. 

Specifically, 
in the case of a time-dependent Hamiltonian, one need not only consider  generators of the Lie group as they appear in $\Ham (t)$, but one can also consider the commuting relations obtained by Magnus expansion:

\begin{proposition}
(i) A necessary condition for reachability of any $\U(T) \in SU(N)$ from  $\U(0) = \hat{I}$ to $\U(T)$, considering Magnus expansion 
\arxiv{(\ref{MagnusExpansion})}
is that the dimension of the Lie algebra generated by Magnus expansion \arxiv{(\ref{MagnusExpansion})} has dimension $N^2 - 1$.
(ii) A necessary and sufficient condition for the existence of time $T^*$ such that for all $T > T^*$ one has exact-time operator controllability 
is that the dimension of the Lie algebra generated by Magnus expansion \arxiv{(\ref{MagnusExpansion})} has dimension $N^2 - 1$.
\label{LieAlgebraDimensionWithMagnus}
\end{proposition}

\arxiv{
\begin{proof}
(i) The key point here is that, via the nested commutation relations, the time-dependent Hamiltonian produces all subalgebras of $\su(N)$ including its Cartan subalgebra. For example, in the case of $SU(4)$ this requires three generators of the Cartan subalgebra given for instance by $T_{\alpha \alpha} = \sigma_\alpha\otimes\sigma_\alpha/2$, where $\sigma_\alpha$ are the Pauli matrices with $\alpha = x,~y,~z$, and at least two generators for each of its two subalgebras  $\su(2)$.
A scheme to perform the Cartan decomposition of the Lie algebra $\su(N)$ for an arbitrary $N$ was introduced in \cite{Su_06}. 
Notice that the condition is necessary but not sufficient in the sense that we do not know the horizon required \emph{a priori}.
(ii) follows from Corollary 4.11 of \cite{Borzi2017}.
\end{proof}
}

The proof is in the Supplementary Material.
It must be pointed out that the validity of our approach is limited by the convergence of the Magnus expansion. This has been  studied extensively\citearxiv{Casas2007,moan2008convergence,lakos2017}, with the most recent result provided by Moan and Niesen \cite{moan2008convergence} in the following proposition:

\begin{proposition}[Moan and Niesen \cite{moan2008convergence}]
\label{convergentMagnus}
The Magnus expansion \arxiv{\eqref{MagnusExpansion}} is absolutely convergent for $t \in  [0,T)$, if
\begin{align}
\label{eq:MagnusConvergenceCriterion}
\int _{0}^{T}\|\hat A(s)\|_{2}\,ds<\pi 
\end{align}
where $\|\cdot \|_{2}$ denotes a matrix norm and 
one assumes the initial condition $\U(0) = \hat{I}$ where $\hat{I}$ is a unit matrix in $\CC^{N \times N}$. 
This result is generic, in the sense that one may construct specific matrices $\hat{A}(t)$ for which the series diverges for any $t > T$.
\end{proposition}




Considering Proposition \ref{convergentMagnus}, we make:

\begin{assumption}[Existence of control] 
\label{Magnus} 
There exist controls 
$\{u_j(t)\} \in \Upsilon$ that attain the minimum of the optimal control problem \eqref{eq:QOC} within  terminal time $T$ satisfying constraint \eqref{eq:MagnusConvergenceCriterion}.
\end{assumption}

\section{The Constructive Argument}

Our second insight is that this approach can be made constructive,
considering that for any number $m$ of terms in the Magnus expansion, 
one obtains a non-commutative polynomial optimisation problem (NCPOP),
which can be solved by solving a sequence \cite[cf.]{Pironio2010} of natural linear
matrix inequalities \cite{boyd1994linear} in the original as well as additional variables for non-linear monomials.

In general, for Hermitian polynomials $p,q_i, i=1,\ldots,m$ over complex numbers, 
NCPOP is:
\begin{align}
\label{polyncprog2}
\min_{(H,X,\phi)}& \langle \phi, p(X) \phi\rangle\\
\text{s.t. }& q_i(X)\succeq
0&\quad i=1,\ldots,m\,,
\notag
\end{align}
where one optimises over  bounded operators
$X=(X_1,\ldots,X_n)
$.
Notice that the 
bounded operator
$p(X)$ acts by  substituting every variable $x_i$ by the operator $X_i$ and every variable $x_i^\dagger$ by $X_i^\dagger$, where $^\dagger$ denotes the adjoint on $\mathcal{H}$.  
If $p^\dagger=p$ is a Hermitian polynomial, then $p(X)=p^\dagger(X)$ is a Hermitian
operator. The notation
$q_i \succeq 0$ suggests operator $q_i$ is positive semi-definite, i.e.,  $\langle \phi, q_i \phi \rangle \geq 0$
for all $\phi\in H$.
The key insight is that a given noncommutative polynomial is positive semidefinite if and only if it decomposes as a sum of Hermitian squares \cite{mccullough2001factorization,helton2002positive}, which is known as the Helton-McCullough Sums of Squares theorem.
For more details, refer to \arxiv{p. \pageref{sec:nc-pop} in the }Supplementary Material.

Using first $\overline m$ terms of the Magnus expansion, the quantum optimal control \eqref{eq:QOC} can be reformulated as\arxiv{ $\mathbf{MEQOC}(\overline m)$}:
\begin{flalign}
\label{eq:QOCNC}
    & \min_{\Omega_m, \{u_j(t)\} \in \Upsilon} J \left(\sum_{m = 1}^{\overline m} \Omega_m (T), \{u_j(t)\} \right)  \\
    & \mbox{s.t. }  \forall m = 1 \ldots \overline m: \notag \\
    & \Omega_m (T) = \sum_i ~\hat{O}_i \notag \\ 
    &\times \int_0^{T} dt_1 \int_0^{t_1} dt_2 
    \dots \int_0^{t_m} dt_{m+1} 
    F_i(\{u_j\}(t_1,\dots,t_{m+1})) \notag 
\end{flalign}
where $\Upsilon$ and $J$ are the same arbitrary semialgebraic set
and a functional as in \eqref{eq:QOC}, respectively.
Notice that $\sum_i \hat{O}_i\tilde{F}_i(\{\tilde{u}_j\}(t_1,\dots,t_{m+1}))$ emerges from the nested commutators of the Magnus expansion which involves the Hamiltonian $\Ham(t) = \sum_j u_j(t) \Ham_j$ at different times.
The operators $\hat{O}_i$ result from the commutators between the Hamiltonian operators $\Ham_j$, and $F_i(\{u_j\}(t_1,\dots,t_{m+1}))$ is a polynomial function of the time-dependent controls $u_j(t)$ at different times, $t_1,\dots,t_{m+1}$, originating from the same commutator.
Notice also that we may need to discretise (subsample) time, as discussed in the Supplementary Material\arxiv{ (p.~\pageref{sec:disc1q})}.



Next, let us ask when there are globally convergent methods for NCPOP.
Considering we can eliminate $\Omega_m$ variables in (\ref{polyncprog2}) by substituting the right-hand side in the objective, 
this concerns largely the additional constraints defining $U$.
Let us define the quadratic module, by following Pironio et al. \cite{Pironio2010}.  Let $Q=\{q_i\}$ be the set of polynomials determining the constraints. 
The \emph{positivity domain} $\mathbf{S}_Q$ of $Q$ are tuples $X=(X_1,\ldots,X_n)$ of bounded operators
 on a Hilbert space $\mathcal{H}$ making all $q_i(X)$ positive semidefinite.
The \emph{quadratic module} $\mathbf{M}_Q$ is the set of  $\sum_if_i^*f_i+\sum_i\sum_j g_{ij}^*q_ig_{ij}$ 
where $f_i$ and $g_{ij}$ are polynomials from the same ring. 
We assume:

\begin{assumption}[Archimedean, cited from \cite{Pironio2010}, cf. also  \cite{burgdorf2016optimization}]
\label{Archimedean}
Quadratic module $\mathbf{M}_Q$ of \eqref{polyncprog2} is Archimedean, i.e., there exists a real constant $C$ such that $C^2-(x_1^*x_1+\cdots+x_{2n}^*x_{2n})\in \mathbf{M}_Q$. 
\end{assumption}

Notice, however, that up to the addition of a redundant constraint, this assumes compactness of the positivity domain, as explained in \cite{Pironio2010}.
The Archimedean assumption is hence very mild.


\begin{figure*}[t!]
\begin{tabular}{llll}
1. Optimal control & 
$\partder{}{t}{} \U (t) = \hat{A}(t) \U (t)$
& \quad &
\\
2. Magnus expansion & 
$\U(T) = \exp \sum_{m = 1}^\infty \Omega_m (T)$
& \quad & \multirow{2}{*}{$\Omega_m (T) = \sum_j ~\hat{O}_j\int_0^{T} dt_1 \int_0^{t_1} dt_2 \dots \int_0^{t_m} dt_{m+1} F(\{u_i\}(t_1,\dots, t_{m+1}))$}
\\
3. Truncated ME & 
$\U(T) \approx \exp \sum_{m = 1}^{\overline m} \Omega_m (T) $
& \quad &
\\
4. Discretised TME & 
$\U(K) \approx \exp \sum_{m = 1}^{\overline m} \tilde \Omega_m (K) $
& \quad &  
$\tilde{\Omega}_m(K) = \sum_j ~\hat{O}_j ~\Delta t^m \sum_{k_1 = 1}^{K} \sum_{k_2 = 1}^{k_1} \dots \sum_{k_{m} = 1}^{k_{m-1}}  ~\tilde{F}_j(\{\tilde{u}_i\}(k_1,\dots,k_m))$
\\
5. Semidefinite program & $\U(K)$ via Theorem \ref{T2} & \quad & 
\end{tabular}
\caption{The steps of our approach outlined. (The approximation in Step 3 is exact in the large limit of $m$. The approximation in Step 4 is exact in the large limit of $m$ and small limit of $\Delta t$.) On the right, we give an impression of the form of the terms of the Magnus expansion, as explained in the Supplementary Material\arxiv{ (p. \pageref{discretisation})}.}
\end{figure*}

\section{Main Result}

With these assumptions, we are ready to state our main result:

\begin{theorem}
\label{T2}
Under Assumptions 
\ref{Magnus} and 
\ref{Archimedean},  
for any initial state $\U(0)$, 
for any lower bound $\underline m$ on the number of terms in the Magnus expansion, 
for any target state in the expanded reachable set ME$(R(\U(0)), \underline m)$,
and any error $\epsilon > 0$, 
there is a number of terms $m(\epsilon) \ge \underline m$ 
such that $\epsilon$-optimal control 
with respect to any polynomially-representable 
functional 
can be extracted from the solution of a certain convex optimisation problem
in the  model of Blum, Shub, and Smale \cite{blum1989}.
\end{theorem}

\arxiv{
\begin{proof}
The proof proceeds in five steps:
first, we need to show that a control exists. 
This is from Assumption \ref{Magnus} and by definition of the reachable set ME$(R(\U(0)), \underline m)$.
Second, we need to show that the Magnus expansion converges. 
This is from Proposition \ref{convergentMagnus}, considering the terminal time $T$ of Assumption \ref{Magnus}.
Third, we need to show that there exists a discretisation that 
 introduces an error of $\delta < \epsilon$.
This could be a uniform discretisation of time with a sufficiently small time step $\Delta t$, as explained in the Supplementary Material\arxiv{ (p. \pageref{discretisation})}. 
Fourth, we need convergence of the series of 
semidefinite-programming (SDP) relaxations
of the discretised non-commutative polynomial optimisation problem.
This is by Theorem 1 of \cite{Pironio2010}, which requires Assumption \ref{Archimedean}; cf. also  \cite{navascues2007bounding,burgdorf2016optimization}, ultimately based on the method of moments \citearxiv{akhiezer1962some}, as explained in the Supplementary Material\arxiv{ (p. \pageref{sec:nc-pop})}.
Finally, we need the extraction of the minimizer from the SDP relaxation of order $r(\epsilon, \delta) \ge \lceil m/2 \rceil$ in the series. 
For this, one utilises the Gelfand--Naimark--Segal (GNS) construction \cite{gelfand1943imbedding,segal1947irreducible}, as explained in Section 2.2 of  \cite{klep2018minimizer}\arxiv{; cf. also  \cite[Section 2.6]{dixmier1969algebres}}.
Alternatively, under restrictive conditions known as the rank loop, we can utilise 
Gram decomposition \cite[cf. proof of Theorem 2]{Pironio2010}.
\end{proof}
}

Notice that we use Magnus expansion in two ways here:
First, $\underline m$ steps in the Magnus expansion guarantee we can reach any target in  $\textrm{ME}(R(\U(0))$. Second, we need $m(\epsilon) \ge \underline m$ number of steps to achieve the convergence within $\epsilon$ error introduced by the Magnus expansion. 
The error of Magnus expansion decays with $B_m/m!$ where $B_m$ is the $m$th Bernoulli number, but it is further compounded by the error in the SDP relaxation of the NCPOP of degree $r(m)$ and the discretisation.
\nonarxiv{
Subsequently, we utilise the Gelfand--Naimark--Segal (GNS) construction \cite{gelfand1943imbedding,segal1947irreducible}, as explained in Section 2.2 of  \cite{klep2018minimizer}.}
The GNS construction is robust to small errors; we can apply \cite[Theorem 3.2]{klep2018minimizer} directly, if there are no constraints $\Upsilon$ \arxiv{\footnote{For other $\Upsilon$s and other constructions \cite{henrion2005detecting}, similar results would have to be developed, yet.}.}

Also notice that in the model of Blum, Shub, and Smale \cite{blum1989},
where basic arithmetic operations with real numbers are atomic, they suggest
not considering floating-point approximation of real numbers and the propagation of the corresponding errors in the analysis.
We discuss the computability and iteration complexity in more detail in the Supplementary Material\arxiv{ (p. \pageref{sec:computability})}.


When we know that the system is operator-controllable, i.e., for any $\U^* \in U(N) \subset \CC^{N \times N}$ there exists a horizon $T$ and an admissible control to drive the initial state $\U(0) = \hat{I}$ to $\U(T) = \U^*$, we can simplify the result as follows:


\begin{coro}
\label{T1}
For an operator-controllable system, 
for any initial state $\U(0)$ and any target state $\U^*$,
and any $\epsilon > 0$,
there exist a number $m$ of terms in the Magnus expansion, 
such that the $\epsilon$-optimal control with respect to fidelity 
can be extracted  
under Assumptions 
\ref{Magnus} and 
\ref{Archimedean}.
\end{coro}

\arxiv{
The assumption of operator controllability, cf. Proposition \ref{LieAlgebraDimensionWithMagnus}, can be relaxed to approximate controllability in the sense of \cite[Chapter 4]{Borzi2017}, i.e., reaching a dense set from any initial state. The analogous corollary follows from  \cite[Theorem 17]{boscain2015approximate}\arxiv{, which in turn is based on the work of Smith \cite{smith1942}}. 
}

We illustrate the results in the Supplementary Material, first on a model of a transmon qubit\arxiv{ (p. \pageref{sec:1q})}, and then on pulse shaping for a two-qubit gate\arxiv{ (p. \pageref{sec:2q})}. In these two examples, there are well-known techniques (starting with the rotating-wave approximation\citearxiv{boscain2004resonance,Magesan_19}) that allow modest numbers of time steps to produce decent discretisations.
Combined with the fidelity of two-qubit gates as a control objective,
this may be the first important application.

\section{Conclusions}

We have presented an approach to quantum optimal control that exhibits global convergence, in theory, and relies on non-trivial but well-developed tools from non-commutative geometry and mathematical optimisation, in practice. In contrast to other quantum control approaches, the use of Magnus expansion provides a proper solution to the initial value problem of the Schr\"odinger equation involving time-dependent Hamiltonian. This has a significant impact on the controllability of quantum systems in that it expands the reachable set both in its dimension and volume. This opens new avenues for research and engineering in quantum control and its applications.

While in the opening paragraph, we have suggested a wide array of possible applications, perhaps the most interesting ones, in the near term, are within quantum computing. 
In Noisy Intermediate Scale Quantum (NISQ) devices, 
fidelity of two-qubit gates is low, but can be improved
with better solvers for quantum optimal control of a 4-level system,
clearly within reach of the proposed method. 
Eventually, one could perhaps replace the application of an entire quantum circuit with the application of a control signal. 
This would make it possible to move beyond the quantum circuit model and the associated intricacy of approximate compiling and swap mapping to accommodate connectivity constraints present in many qubit technologies.

Globally convergent methods for quantum optimal control also make it possible to study the inherent limits of quantum systems. The specific objective could be to map the parameter landscape of a specific physical platform \cite{goerz2017charting} such as a transmon qubit\arxiv{ \cite{Magesan_19}}, and thus identify efficient operating regimes, for example. 
Indeed, the use of optimal-control methods that are not globally convergent
maps the parameter landscape reflecting both the particular combination of an optimal-control method as well as the underlying physical platform, rather than just the platform itself. 
Perhaps similar to \cite{nielsen2006quantum},
our method therefore enables a more profound understanding of quantum systems. 

\paragraph*{Acknowledgments}
\begin{acknowledgments}
The authors would like to thank Jason Crain, Daniel Egger, Didier Henrion, Christiane Koch, Vyacheslav Kungurtsev, and Monique Laurent for comments that have improved the presentation of the work.
This work has been supported by European Union’s Horizon 2020 research and innovation programme under the Marie Skłodowska-Curie Actions, grant agreement 813211 (POEMA).
Authors' contributions: J.M. and J. V. designed the research and performed 
the formal analysis and wrote the original draft.
Competing interests: no competing interests.
Data and materials availability: all data is available in the manuscript or the Supplementary Materials.
\end{acknowledgments}

\arxiv{
\clearpage
\appendix
\onecolumngrid
\section{Supplementary Material}

\subsection{Functionals}
\label{sec:functionals}

Our results on the quantum optimal control:

\begin{align}
    \displaystyle\min_{\U(t), \{u_j(t)\} \in \Upsilon} & J \left(\U (t), \{u_j(t)\} \right) \notag \\
    \mbox{s.t.} ~~&\partder{}{t}{} \U (t) =  \left[\sum_j u_j(t) ~\Ham_j/i\hbar\right] \U (t), \non \\
    &\U(0) = \hat{I}
    \non
\end{align}

are applicable to a wide-range of functionals $J$, as long as $J$ can be represented by a polynomial in the original and additional scalar and matrix variables, subject to some matrices or submatrices being positive semidefinite, and subject to polynomial constraints either in scalar or matricial variables. This extends the notion of semidefinite representability \citearxiv{bental2001lectures,Sagnol2013} and the notion of polynomial representability.

\paragraph{Textbook state-related functionals}
Typically, one considers functionals wherein at least one summand carries some notion of similarity of the terminal state $\hat{\rho}(T)$ and the target state $\hat{\rho}^*$.
The proper distance measure is provided by the trace distance defined as
\begin{align}
\label{eq:trace}
    \textrm{trace distance} \quad & \frac{1}{2} \tr \abs{\hat{\rho}(T) - \hat{\rho}^*}
\end{align}
where $\abs{A} = \sqrt{A^\dagger A}$. The trace distance generalizes the Kolmogorov distance between two probability distributions \cite{Nielsen_00}.
A variety of matrix norms can also be considered \cite{Horn_85}. A particularly useful functional in the context of unitary operations can be derived from the Frobenius norm:
\begin{align}
\label{eq:frob}
    \textrm{Frobenius functional} \quad & \sqrt{\tr \left(\U^{*\dagger}\U(T)\right)}.
\end{align}
This functional can easily be augmented by the addition of other terms which are used to control leakage of the quantum population to the rest of a larger Hilbert space and other undesired processes.

Alternatively, one may consider fidelity, which is defined as follows:
\begin{align}
\textrm{fidelity } \quad & \tr \sqrt{{\sqrt {\hat{\rho}^*}} ~\hat{\rho}(T) ~{\sqrt {\hat{\rho}^*}}}
\end{align}
where both $\hat{\rho}(T)$ and $\hat{\rho}^*$ are mixed states in general. If the target state is pure, i.e. $\hat{\rho}^* = \ket{\psi}\bra{\psi}$, it is easy to show that the fidelity simplifies to $F(\hat{\rho}(T),\hat{\rho}^*) = \sqrt{\bra{\psi} \hat{\rho}(T) \ket{\psi}}$. 
The fidelity is not a proper distance measure, but it satisfies several similar properties: It is invariant under unitary transforms and symmetric in its inputs $F(\hat{\rho},\hat{\sigma}) = F(\hat{\sigma},\hat{\rho})$. Its values are $0 \ge F(\hat{\rho},\hat{\sigma}) \ge 1$
where $F(\hat{\rho},\hat{\sigma}) = 0$ if the states $\hat{\rho}$ and $\hat{\sigma}$ have support on orthogonal subspaces, and $F(\hat{\rho},\hat{\sigma}) = 1$ if and only if $\hat{\rho} = \hat{\sigma}$ \cite{Nielsen_00}.


\newpage


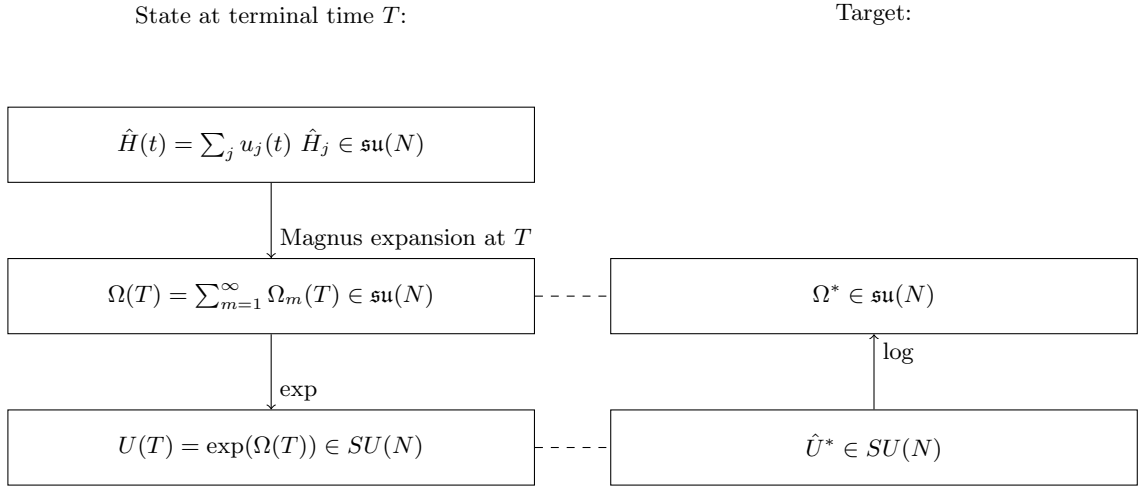
\begin{figure*}[htb]
\begin{tikzpicture}
\tikzset{edge/.style = {->,> = latex'}}
\node [draw, minimum width=7cm, minimum height = 1cm] (A1) {$\Omega(T) = \sum_{m = 1}^{\infty} \Omega_m (T) \in \su(N)$};

\node [draw, minimum width=7cm, minimum height = 1cm, above=of A1](H){$\hat{H}(t) = \sum_j u_j(t) ~\Ham_j \in \su(N)$};
\node [minimum width=7cm, minimum height = 1cm, right=of H] (Blank){};

\node [minimum width=7cm, above=of H](Header1){State at terminal time $T$:};
\node [minimum width=7cm, right=of Header1] (Header2){Target:};

\node [draw, minimum width=7cm, minimum height = 1cm, below=of A1](B1){$U(T) = \exp(\Omega(T)) \in SU(N)$};

\node [draw, minimum width=7cm, minimum height = 1cm, right=of A1] (A2){$\Omega^* \in \su(N)$};
\node [draw, minimum width=7cm, minimum height = 1cm, below=of A2](B2){$\U^* \in SU(N)$};

\draw [->, to path={-| (\tikztotarget)}] (H.south) to (A1.north)  node [above right] {Magnus expansion at $T$};

\draw [dashed] (A1.east) edge (A2.west);
\draw [->, to path={-| (\tikztotarget)}] (A1.south) to (B1.north)  node [above right] {$\exp$};
\draw [dashed] (B1.east) edge (B2.west);
\draw [->, to path={-| (\tikztotarget)}] (B2.north) to (A2.south)  node [below right] {$\log$};
\end{tikzpicture}
\caption{The ``algebra'' and ``group'' view of the quantum optimal control problem \ref{eq:QOC}.}
\label{fig:twolevels}
\end{figure*}

\paragraph{State-related terms implemented}

Notice that in \eqref{eq:QOCNC}, we actually use $\sum_{m = 1}^{\overline m} \Omega_m (T)$ as the first argument of the functional,
rather than $\U(t)$. 
Trivially, one could argue that from the Cayley–Hamilton theorem, the matrix exponential is polynomially representable. Indeed, from the power-series expansion:
\begin{align}
\exp (\Omega) = \sum_{k=0}^\infty\frac{\Omega^k}{k!} = I + \Omega + \frac{1}{2}\Omega^2 + \frac{1}{6}\Omega^3 + \cdots,
\end{align}
but this would not allow for an efficient implementation. 

An efficient implementation requires some additional work, relating the algebra $\su(N)$ and the group $SU(N)$, as suggested in Figure \ref{fig:twolevels}, with details dependent on the functional. 
In the trace distance, the most efficient option is to pre-compute the matrix logarithm of $\hat{\rho}^*$, numerically \cite[Chapter 11]{higham2008functions}, and then compare directly against  
$\sum_{m = 1}^{\overline m} \Omega_m (T)$. In order to do this, we have to relate the minimiser of the trace distance and the exponential of a minimiser of the trace of the original variable, as suggested in Figure \ref{fig:twolevels}, while respecting periodicity of the unitary operators as functions of the algebra elements:

\begin{proposition}[]
Let $\arg \min_{\Omega(T) \in \su(N)} \tr\abs{\Omega(T) - \log \U^*} $ be the global minimiser of $\Omega(T)$ with the global minimum being $0$.
Let $\arg \min_{\exp(\Omega(T)) \in SU(N)} \tr(\exp\abs{\Omega(T)) - \U^*}$ be the set of global minimisers of $U(T)$ with the global minimum being 0.
Then, 
\begin{align}
\exp \left( \arg \min_{\Omega(T) \in \su(N)} \tr\abs{\Omega(T) - \log \U^*} \right)  \in \arg \min_{\exp(\Omega(T)) \in SU(N)} \tr\abs{\exp(\Omega(T)) - \U^*},
\end{align}
where $\abs{A} = \sqrt{\hat{A}^\dagger \hat{A}}$.
\end{proposition}

\begin{proof}
When the trace distance $\tr\abs{\U - \U^*}$ between an element of the unitary operator $\U = \exp(\Omega(T))$, produced by an optimal control procedure, and the operator $\U^*$, representing the target, is zero, then $\U = \exp(\Omega)= \exp(\Omega^*) = \U^*$ and hence $\Omega = \Omega^*$. In both cases, $\exp(\Omega) ~\exp(-\Omega) = \exp(-\Omega) ~\exp(\Omega) = \hat{I}$ where $\hat{I}$ is an identity operator.  
\end{proof}





Similar reasoning can be applied to other distances $D$:

\begin{proposition}[]
Let $\arg \min_{\Omega(T) \in \su(N)} D(\Omega(T), \log \U^*) $ be the global minimiser of $\Omega(T)$ with the global minimum being $0$.
Let $\arg \min_{\exp(\Omega(T)) \in SU(N)} D(\exp(\Omega(T)) - \U^*)$ be the set of global minimisers of $U(T)$ with the global minimum being 0.
Then, 
\begin{align}
\exp \left( \arg \min_{\Omega(T) \in \su(N)} D(\Omega(T),  \log \U^*) \right)  \in \arg \min_{\exp(\Omega(T)) \in SU(N)} D(\exp(\Omega(T)), \U^*).
\end{align}
\end{proposition}

\begin{proof}
The proof follows the same steps as above. When $D( \U, \U^*) = \tr\abs{\U - \U^*}$ between an element of the unitary operator $\U = \exp(\Omega(T))$, produced by an optimal control procedure, and the operator $\U^*$, representing the target, is zero, then $\U = \exp(\Omega)= \exp(\Omega^*) = \U^*$ and hence $\Omega = \Omega^*$. In both cases, $\exp(\Omega) ~\exp(-\Omega) = \exp(-\Omega) ~\exp(\Omega) = \hat{I}$ where $\hat{I}$ is an identity operator.
\end{proof}

Alternatively, one should realise that 
while a matrix exponential is not an operator monotone, its trace is a monotone trace function \cite[Section 8.3.2 Monotone Trace Functions]{tropp2015introduction}:

\begin{proposition}[Theorem 5 in \cite{Fawzi2019}]
Let $g$ be a monotone trace function and let $a > 1$. Then for any $\epsilon > 0$ there
is a rational function $r$ such that $|r(x) - g(x)| \le \epsilon$ for all $x \in [1/a, a]$, and $r$ has a semidefinite
representation of size $O(\log(1/\epsilon))$.
\end{proposition}

A matrix logarithm is also semidefinite-representable \cite{Fawzi2019}, which makes it possible to compare against $\log(\U^*)$ even if $\U^*$ were a variable, for instance within robust optimisation.


\paragraph{Control-related terms}

Alternatively, 
fidelity can be used as a constraint, rather than a part of the objective.
In the objective, one could consider the energy of the control signal in a functional such as:
\begin{align}
\textrm{non-isotropic energy } \quad & \sum_{t=1}^{T} \sum_j  \frac{1}{\mu_j} \left | u_j(t) \right |^2,\\
\textrm{non-isotropic length } \quad & \sum_{t=1}^{T} \sqrt { \sum_j \frac{1}{\mu_j} \left | u_j(t) \right |^2} ,\\
\textrm{non-isotropic area } \quad & \sum_{t=1}^{T} \sum_j \frac{1}{\mu_j} \left | u_j(t) \right |,
\end{align}
for some scalar constants $\mu_j$.
In this case, there are obviously no issues with representability of the functional.





\clearpage
\section{Discretization}
\label{discretisation}

In order to evaluate the Magnus expansion numerically, we consider discrete representation of all functions of time and a suitable quadrature for evaluation of the definite integrals for the total time $T$ divided into $K$ knots $t_k$ \cite{singh2018high} 
\begin{align}
    \int_0^T ~f(t)~dt \approx \sum_{k=1}^K~w_kf(t_k) \non
\end{align}
where $w_k$ are the quadrature weights.

Formally, we perform the discretisation as follows:
\begin{align}
    \hat{U}(T) &= \mbox{exp}\left(\sum_{m=1}^{\infty} ~\Omega_m(t)\right) \rightarrow \exp\left(\sum_{m=1}^{\infty} ~\hat{\Omega}_m(t)\right)
\end{align}
where we have chosen  $t_k = k~\Delta t$ for $k = 1, \dots, K$. The individual terms of the Magnus expansion transform as 
\begin{align}
    \Omega_m (T) &= \sum_i ~\hat{O}_i\int_0^{T} dt_1 \int_0^{t_1} dt_2 \dots \int_0^{t_m} dt_{m+1} F_i(\{u_j\}(t_1,\dots,t_{m+1})) \non \\
    \rightarrow 
    \tilde{\Omega}_m (T) &= \sum_i ~\hat{O}_i ~\Delta t^m \sum_{k_1 = 1}^{K} \sum_{k_2 = 1}^{k_1} \dots \sum_{k_{m} = 1}^{k_{m-1}}  ~\tilde{F}_i(\{\tilde{u}_j\}(k_1,\dots,k_m)) 
\end{align}
where $\sum_i \hat{O}_i\tilde{F}_i(\{\tilde{u}_j\}(t_1,\dots,t_{m+1}))$ emerges from the nested commutators of the Magnus expansion which involves the Hamiltonian $\Ham(t) = \sum_j u_j(t) \Ham_j$ at different times.
The operators $\hat{O}_i$ result from the commutators between the Hamiltonian operators $\Ham_j$, and $F_i(\{u_j\}(t_1,\dots,t_{m+1}))$ is a polynomial function of the time-dependent controls $u_j(t)$ at different times, $t_1,\dots,t_{m+1}$, originating from the same commutator.
For example, at first order, these functions are the controls themselves, at the second order they are products of controls relevant to different noncommuting terms of the Hamiltonian $\Ham_j$ at time $t_1$ and $t_2$, and so on. Furthermore $\tilde{\Omega}_m$, $\tilde{F}$ and $\tilde{u}$ denote appropriate entities in their discretised form.

Explicitely
\begin{align}
\tilde{\Omega}_1(T) &= \Delta t~\sum_{k_1 = 1}^K ~\tilde{\hat{A}}(k_1), \non \\
\tilde{\Omega}_2(T) &= \frac{1}{2} ~\Delta t^2 ~\sum_{k_{1} = 1}^{K}   ~\sum_{k_{2} = 1}^{k_{1}} ~\comm{\tilde{\hat{A}}(k_{1})}{\tilde{\hat{A}}(k_{2})}, \non \\
\tilde{\Omega}_3(T) &= \frac{1}{6} ~\Delta t^3 ~\sum_{k_{1} = 1}^{K}   ~\sum_{k_{2} = 1}^{k_{1}} ~\sum_{k_{3} = 1}^{k_{2}} \left(\comm{\tilde{\hat{A}}(k_{1})}{\comm{\tilde{\hat{A}}(k_{2})}{\tilde{\hat{A}}(k_{3})}} + \comm{\comm{\tilde{\hat{A}}(k_{1})}{\tilde{\hat{A}}(k_{2})}}{\tilde{\hat{A}}(k_{3})} \right), \non \\
& \dots \non
\end{align}

where the discretised operators 
\begin{align}
    \tilde{\hat{A}}(k) = \sum_j \tilde{u}_j(k)~\Ham_j / i\hbar \non
\end{align}

and the nested commutators lead to the following structure of the Magnus operators
\begin{align}
\tilde{\Omega}_1(T) &= \frac{\Delta t}{i \hbar} ~\sum_{j}~\Ham_j~\sum_{k_1 = 1}^K ~\tilde{u}_j(k_1) = ~\sum_j ~\hat{O}_j^{(1)} ~\Delta t \sum_{k_1 = 1}^{K} ~\tilde{F}_j^{(1)}(\{\tilde{u}_j\}(k_1)) 
, \non \\
\tilde{\Omega}_2(T) &= \frac{1}{2} \frac{\Delta t^2}{(i\hbar)^2} ~\sum_{j_1,j_2} 
~\comm{\Ham_{j_1}}{\Ham_{j_2}}
~\sum_{k_{1} = 1}^{K}   ~\sum_{k_{2} = 1}^{k_{1}} ~u_{j_1}(k_1)~u_{j_2}(k_2)\non \\
&= ~\sum_j ~\hat{O}_j^{(2)} ~\Delta t^2 \sum_{k_1 = 1}^{K} ~\sum_{k_{2} = 1}^{k_{1}} ~\tilde{F}_j^{(2)}(\{\tilde{u}_j\}(k_1,k_2)) 
, \non \\
\tilde{\Omega}_3(T) &= \frac{1}{6} \frac{\Delta t^3}{(i\hbar)^3}
~\sum_{j_1,j_2,j_3} 
\left(\comm{\Ham_{j_1}}{\comm{\Ham_{j_2}}{\Ham_{j_3}}} +
~\comm{\comm{\Ham_{j_1}}{\Ham_{j_2}}}{\Ham_{j_3}}\right)
~\sum_{k_{1} = 1}^{K}   ~\sum_{k_{2} = 1}^{k_{1}} ~\sum_{k_{3} = 1}^{k_{2}} ~u_{j_1}(k_1)~u_{j_2}(k_2)~u_{j_3}(k_3) \non \\
&= ~\sum_j ~\hat{O}_j^{(3)} ~\Delta t^3 ~\sum_{k_{1} = 1}^{K}   ~\sum_{k_{2} = 1}^{k_{1}} ~\sum_{k_{3} = 1}^{k_{2}} 
~\tilde{F}_j^{(3)}(\{\tilde{u}_j\}(k_1,k_2,k_3)), \non \\
\dots \non
\end{align}
where the operators $\hat{O}_j^{(k)}$ result from the nested commutators of the $k$-th term of the Magnus expansion, and $\tilde{F}_j^{(k)}$ are the  corresponding functions of the set of controls $\{u\}_j$ at relevant times labelled by discrete indices $k_1, k_2, k_3, \dots$~.

\clearpage 
\subsection{One qubit}
\label{sec:1q}

Let us illustrate our approach on the simplest possible example. 
We consider a quantum control of a single qubit where the control target is a unitary transformation
$\U^* \in SU(2)$. 
Furthermore, we consider the time dependent Hamiltonian for a spin 1/2 system or one quantum bit of the following form
\begin{align}
\Ham(t) = a(t) \Sz + b(t) \Sx 
\end{align}
where $a(t)$ and $b(t)$ are scalar functions of time, $\Sz = \hbar \sigma_z / 2$ and  $\Sx = \hbar \sigma_x / 2$ are the operators for the $z$-component and $x$-component of the spin angular momentum respectively, and $\sigma_z$ and $\sigma_x$ are the Pauli matrices. The components of the spin angular momentum satisfy the following commutation relations
\begin{align}
\comm{\Sx}{\Sy} = i \hbar \Sz, ~~~~~\comm{\Sy}{\Sz} = i \hbar \Sx, ~~~~~\comm{\Sz}{\Sx} = i \hbar \Sy. 
\end{align}

These commutation relations have an important consequence in that the Hamiltonian does not in general commute with itself at different times
\begin{align} \label{single_qubit}
\comm{\Ham(t_1)}{\Ham(t_2)} = 
i \hbar \left[a(t_1) b(t_2) - a(t_2) b(t_1)\right] \Sy 
\end{align}
provided $a(t)$ and $b(t)$ are not constant in time or are not proportional to each other $a(t) = c b(t)$ where $c$ is a real constant. \\

The individual terms of the Magnus expansion at low orders are then given as
\begin{align}
\Omega_1 (T) &= \frac{1}{i \hbar} \int_{0}^{T} dt_1~\Ham(t_1) = \frac{1}{i \hbar} \Sz \int_{0}^{T} dt_1~a(t_1) + \frac{1}{i \hbar} \Sx \int_{0}^{T} dt_1~b(t_1), \non \\
\Omega_2 (T) &= \frac{1}{i \hbar} \frac{1}{2} ~\Sy ~\int_{0}^{T} dt_1 \int_{0}^{t_1} dt_2 \left[a(t_1) b(t_2) - a(t_2) b(t_1)\right],  \non \\
\Omega_3 (T) &= \frac{1}{i \hbar} \frac{1}{6} \Sx ~\int_{0}^{T} dt_1 \int_{0}^{t_1} dt_2  \int_{0}^{t_2} dt_3 \non \\ 
&\times \left[ a(t_1) a(t_2) b(t_3) - a(t_1) b(t_2) a(t_3) + b(t_1) a(t_2) a(t_3)  - a(t_1) b(t_2) a(t_3) \right] \non \\
&- \frac{1}{i \hbar} \frac{1}{6} \Sz ~\int_{0}^{T} dt_1 \int_{0}^{t_1} dt_2  \int_{0}^{t_2} dt_3 \non \\ 
&\times \left[ b(t_1) a(t_2) b(t_3) - b(t_1) b(t_2) a(t_3) + b(t_1) a(t_2) b(t_3)  - a(t_1) b(t_2) b(t_3) \right] .
\end{align}

\subsection{System with a constant drift and time-dependent driving}

The situation when both coefficients in the Hamiltonian Eq. (\ref{single_qubit}) are time dependent is the most general but perhaps too complicated if one considers experimental realisations. It seems well motivated to consider a Hamiltonian which consists of a constant drift term, given by $a(t) = a$, and a time-dependent driving term 
\begin{align}
\Ham(t) = a \Sz + b(t) \Sx.
\end{align} 
In this case the terms of the Magnus expansion above simplify further yielding
\begin{align}
\Omega_1 (T) &= \frac{1}{i \hbar} a T \Sz + \frac{1}{i \hbar} \Sx \int_{0}^{T} dt_1 ~b(t_1), \non \\
\Omega_2 (T) &= \frac{1}{i \hbar} \frac{a}{2} \Sy \left[ \int_{0}^{T} dt_1 \int_{0}^{t_1} dt_2 ~b(t_2) - \int_{0}^{T} dt_1 ~t_1~b(t_1)\right] \non \\
\Omega_3 (T) &= \frac{1}{i \hbar} \frac{a^2}{6} \Sx \int_{0}^{T} dt_1 \int_{0}^{t_1} dt_2 \int_{0}^{t_2} dt_3 \left[b(t_3) - 2 b(t_2) + b(t_1)\right] \non \\
& - \frac{1}{i \hbar} \frac{a}{6} \Sz \int_{0}^{T} dt_1 \int_{0}^{t_1} dt_2 \int_{0}^{t_2} dt_3 \left[2 b(t_1) b(t_3) - b(t_1) b(t_2) - b(t_2) b(t_3)\right]. 
\label{Magnus_terms}
\end{align}

\subsubsection{Exactly solvable case: linear driving}

In the case that the driving term in the Hamiltonian is the function $b(t) = b t$ for some real constant $b$ then the terms of the Magnus expansion for $n > 1$ can be evaluated analytically \cite{Salzman_87}, yielding the following expressions
\begin{align}
\Omega_{2n + 1} (T) &= 0, \non \\
\Omega_{2n} (T) &= \frac{(-1)^n}{i \hbar} ~ T^{2n+1} ~a^{2n - 1}~b ~B_{2n} ~\Sy 
\end{align}
where $B_{2n}$ are related to Bernoulli numbers and are given for the low order as
\begin{align}
B_2 &= \frac{1}{12}, \non \\
B_4 &= - \frac{1}{720}, \non \\
B_6 &= \frac{1}{30240}, \non \\
B_8 &= -\frac{1}{1209600}, \non \\
B_{10} &= \frac{1}{47900160}. 
\end{align}

The formulas for the Magnus terms in the case $a(t) = a$ and $b(t) = b t$ give explicitly 
\begin{align}
\Omega(T) &= \sum_{k = 1}^{\infty} \Omega_k(T)  \non \\
&= \frac{a}{i \hbar} T \Sz + \frac{1}{i \hbar} \frac{b}{2} T^2 \Sx + \frac{b}{i \hbar} \left[\frac{a}{12} T^3 + \frac{a^3}{720}  T^5 + \frac{a^5}{30240} T^7 + \frac{a^7}{1209600} T^9 + \frac{a^9}{47900160} T^{11} + \dots \right] \Sy. \non \\
&= -i ~\left(\theta_x(T) \sigma_x + \theta_y(T) \sigma_y + \theta_z(T)  \sigma_z\right) \non \\
&= -i \abs{\vec{\theta}(T)} ~\vec{\Theta(T)} \cdot \vec{\sigma} 
\end{align}
where $\abs{\vec{\theta}(T)} = \sqrt{\theta_x(T)^2 + \theta_y(T)^2 + \theta_z(T)^2}$ and  $\Theta_\alpha(T) = \theta_\alpha(T) / \sqrt{\theta_x(T)^2 + \theta_y(T)^2 + \theta_z(T)^2}$, $\alpha = x, y ,z$, are the components of a unit vector $\vec{\Theta}(T)$. \\

In this representation, we can write down the evolution operator explicitly as
\begin{align}
\U(T) = \exp \Omega(T)&= \exp  \left[-i \abs{\vec{\theta}(T)} ~\vec{\Theta}(T) \cdot \vec{\sigma}\right] \non \\
&= \cos{\abs{\vec{\theta}(T)}} ~~ \hat{I} - i \sin{\abs{\vec{\theta}(T)}} ~~\vec{\Theta}(T) \cdot \vec{\sigma} \non \\
&= \begin{pmatrix}
c_T - i s_T ~\Theta_z & -s_T (~\Theta_y + i ~\Theta_x) \\
s_T (~\Theta_y - i ~\Theta_x) & c_T + i s_T ~\Theta_z.
\end{pmatrix}
\end{align}
where we simplified the notation by substituting $\sin\abs{\vec{\theta}(T)}$ with a real variable $s_T$ and replaced $\cos\abs{\vec{\theta}(T)}$ with a real variable $c_T$. 

In the case of linear driving, Salzman \cite{Salzman_87}
observed that the ratio $B_{n+2}/B_{n}$ converges to $-1/2\pi$. This implies that the ratio $\Omega_{n+2}/\Omega_n$ converges to $\left(a T/2 \pi\right)^2$ and the criterium $a T/2\pi < 1$ for the convergence of the Magnus expansion in this particular case.

\subsubsection{Discretization}
\label{sec:disc1q}

We intend to use the analytical results of the previous section to evaluate error associated with the discretized version of the Magnus expansion that is required for any numerical calculations. In the case of a Hamiltonian with a constant drift term, this means in the simplest instance replacing the integrals in Eq.~(\ref{Magnus_terms}) with sums, whereas the final time is given as $T = k \Delta t$ for $k = K$:

\begin{align}
\Omega_1 (T) &= \frac{1}{i \hbar} a ~K \Delta t ~\Sz + \frac{1}{i \hbar} ~\Delta t ~\Sx \sum_{k=1}^{K} ~b(k), \non \\
\Omega_2 (T) &= \frac{1}{i \hbar} \frac{a}{2} ~\Delta t^2~\Sy \left[\sum_{k_1=1}^{K} \sum_{k_2=1}^{k_1} ~b(k_2) -  \sum_{k_1=1}^{K} ~k_1~b(k_1)\right] \non \\
\Omega_3 (T) &= \frac{1}{i \hbar} \frac{a^2}{6} ~\Delta t^3 ~\Sx
\left[\sum_{k_1=1}^{K}\sum_{k_2=1}^{k_1}\sum_{k_3=1}^{k_2}
~b(k_3) - 
2 \sum_{k_1=1}^{K}\sum_{k_2=1}^{k_1}
~k_2~b(k_2) +
\frac{1}{2}\sum_{k_1=1}^{K}
~k_1 (k_1 + 1)~b(k_1)\right] \non \\
& - \frac{1}{i \hbar} \frac{a}{6}~\Delta t^3 \Sz
\left[2 \sum_{k_1=1}^{K}~b(k_1)\sum_{k_2=1}^{k_1}\sum_{k_3=1}^{k_2}
~b(k_3) -
\sum_{k_1=1}^{K}~b(k_1)\sum_{k_2=1}^{k_1}
k_2~b(k_2) -
\sum_{k_1=1}^{K}\sum_{k_2=1}^{k_1}~b(k_2)\sum_{k_3=1}^{k_2}
~b(k_3)\right]
\end{align}
and the discrete values $b(k)$ are taken at the midpoint, i.e. $b(k) = b((k - 1/2)\Delta t)$, to maintain accuracy of the discrete representation of the integrals which appear in Eq.~(\ref{Magnus_terms}).

\clearpage
\subsection{Two qubits}
\label{sec:2q}

Elements of the group $SU(4)$ of two-qubit operations can be written as a complex exponential function of the $\su(4)$ algebra generators. These naturally split into three sets: (i) $\T_{\alpha 0} = \sigma_{\alpha} \otimes \hat{I}/2$, (ii) $\T_{0 \alpha} = \hat{I} \otimes \sigma_{\alpha}/2$, and (iii) $\T_{\alpha \beta} = \sigma_{\beta} \otimes \sigma_{\alpha}/2$, where $\hat{I}$ is a $2 \times 2$ unit matrix and $\sigma_{\alpha}$, $\alpha = x,~y,~z$, are the Pauli matrices. The generators $\T_{\alpha 0}$ and $\T_{0 \alpha}$ correspond to single qubit operations on the first and second qubit respectively and hence generate the subgroup $SU(2) \otimes SU(2) \in SU(4)$.

Let us consider an envelope of the electromagnetic field, $\omega(t)$,
which we use to control a system of two qubits.
There, the only rigorous methods are based on \cite{d2000algorithms,khaneja2001time,981724}, which are not constructive.

We consider the system of two transmon qubits coupled to a bus resonator \cite{Magesan_19} which is characterized in the Rotating Wave Approximation (RWA) by the following Hamiltonian

\begin{align}
\Ham & = \Ham_0 + \Omega(t) \Ham(t)  = \hbar
\begin{pmatrix}
0 & 0 & \Omega(t)/2 & 0 \\
0 & 0 & J & \Omega(t)/2 \\
\Omega(t)/2 & J & \Delta & 0 \\
0 & \Omega(t)/2 & 0 & \Delta 
\end{pmatrix}\non \\
& = \frac{\hbar\Delta}{2} \hat{I} - \frac{\hbar\Delta}{2} ~\sigma_z \otimes \hat I + \frac{\hbar\Omega(t)}{2}~\sigma_x \otimes \hat{I} + \frac{\hbar J}{2} ~\left(\sigma_x \otimes \sigma_x + \sigma_y \otimes \sigma_y \right)
\label{2qubits}
\end{align}
where $\Delta = \bar{\omega}_1 - \bar{\omega}_2$ is a detuning between transition frequencies of the individual transmon qubits dressed by the resonator frequency, $J$ is the exchange coupling and $\Omega$ is the Rabi frequency  which drives the control qubit at the frequency of the target qubit to implement a two-qubit operation. The term $\Delta \hat{I}/2$ in Eq (\ref{2qubits}) generates a global phase in quantum dynamics and can be ignored, so the operator $\hat{A}(t) = \Ham/i\hbar$ can be written in terms of the $\su(4)$ algebra generators
\begin{align}
\hat{A} & = - \frac{\Delta}{2i} ~\sigma_z \otimes \hat I + \frac{\Omega(t)}{2i}~\sigma_x \otimes \hat{I} + \frac{J}{2i} ~\left(\sigma_x \otimes \sigma_x + \sigma_y \otimes \sigma_y \right) \non \\
&= \frac{1}{i} \left[- \Delta ~\T_{z0} + \Omega(t)~\T_{x0} + J~\left(\T_{xx} + \T_{yy}\right)\right].
\label{2qubits_2}
\end{align}
We emphasize that the Rabi frequency is time dependent $\Omega (t)$ and is used to control quantum dynamics of the system in order to steer the system to the desired control target like a specific unitary two-qubit operation. This significantly differs from the proposal \cite{Magesan_19} which considers the Rabi frequency $\Omega$ to be constant. 

We are now in a position to write down the individual terms of the Magnus expansion

\begin{align}
\Omega_1(T) &= \frac{1}{i}\int_0^T dt_1~\hat{H}(t_1), = \frac{1}{i} \left[-\Delta \T_{z0} + J\left(\T_{xx} + \T_{yy}\right)\right]~T + \frac{1}{i}~\T_{x0}~\int_0^T dt ~\Omega(t),\non \\
\Omega_2(T) &= \frac{1}{2i} \left(\Delta~\T_{y0}+ J~\T_{zy}\right)\left[\int_0^T  dt_1 ~t_1~\Omega(t_1) - \int_0^T  dt_1  \int_0^{t_1}  dt_2~\Omega(t_2)\right], \non \\
\Omega_3(T) &= \frac{1}{6i} 
\left\{
\left[\left(\Delta^2 + J^2\right)~\T_{x0} + J \Delta~\T_{zx}\right] \times \right. \non \\
&\left[ 2 \int_0^T  dt_1 \int_0^{t_1}  dt_2 ~t_2~\Omega(t_2) 
- \frac{1}{2} \int_0^T  dt_1 ~t_1^2~\Omega(t_1)
- \int_0^T  dt_1 \int_0^{t_1}  dt_2 \int_0^{t_2}  dt_3 ~\Omega(t_3)\right]\non \\
& + \left(\Delta~\T_{z0} - J~\T_{yy}\right)\times \left[\int_0^T  dt_1 ~\Omega(t_1) \int_0^{t_1}  dt_2  ~t_2~\Omega(t_2)\right. \non \\
&\left. \left. -2 \int_0^T  dt_1 ~\Omega(t_1) \int_0^{t_1}  dt_2  \int_0^{t_2}  dt_3 ~\Omega(t_3) + \int_0^T  dt_1  \int_0^{t_1}  dt_2 ~\Omega(t_2) \int_0^{t_2}  dt_3 ~\Omega(t_3)\right]\right\}, \non \\
\dots \non~.
\end{align}

\subsection{Controllability}

We observe that the Hamiltonian
\begin{align}
\hat{H} & = \left[- \Delta ~\T_{z0} + \Omega(t)~\T_{x0} + J~\left(\T_{xx} + \T_{yy}\right)\right] \non
\end{align}
produces, via the nested commutators of the Magnus expansion, the terms involving the following generators of the algebra $\su(4)$:
\begin{align}
&\left\{\T_{x0}, \T_{z0}, \T_{xx}, \T_{yy}\right\}^{k=1} \rightarrow 
\left\{\T_{y0}, \T_{zy} \right\}^{k=2} \rightarrow
\left\{\T_{x0}, \T_{z0}, \T_{zx}, \T_{yy}\right\}^{k=3} \rightarrow
\left\{\T_{y0}, \T_{xy}, \T_{yx}, \T_{zy}\right\}^{k=4} \non \\ &\rightarrow \cdots ~.
\non
\end{align} 

This set of generators is closed upon the repeated commutation with the Hamiltonian and thus this system is not operator controllable in the sense defined above.

\subsection{Towards operator controllability}

We notice that the set of generators obtained via the Magnus expansion contains all the generators for the algebra $\su(2)$ on the first qubit, and six two-qubit generators: \\
$\T_{xx}$, $\T_{xy}$, $\T_{yx}$, $\T_{yy}$, $\T_{zx}$, $\T_{zy}$.

The generators which are not obtained include those of $\su(2)$ on the second qubit:
$\T_{0x}$, $\T_{0y}$, and $\T_{0z}$,
and the two qubit generators: 
$\T_{xz}$, $\T_{yz}$, and $\T_{zz}$.

The Cartan decomposition of $\su(4)$ splits the algebra into two parts $\mathfrak{p} \oplus \mathfrak{k}$ where 
$\mathfrak{p}$ is the Cartan subalgebra generated by 
$\T_{xx}$, $\T_{yy}$, and $\T_{zz}$, and $\mathfrak{k}$
which consists of single qubit subalgebras.

They satisfy the commutation relations
\begin{align}
    \comm{\mathfrak{k}}{\mathfrak{k}} = \mathfrak{k},~~~~~
    \comm{\mathfrak{p}}{\mathfrak{k}} = \mathfrak{p},~~~~~
    \comm{\mathfrak{p}}{\mathfrak{p}} = \mathfrak{k}\non
\end{align}
where the last one suggests that $\su(2)$ on the second qubit can be obtained if our set of generators also contains $\T_{xz}$, $\T_{yz}$, and $\T_{zz}$. 

The operator controllability is achieved by augmenting the Hamiltonian by the generator $\T_{zz}$.
In superconducting qubits, $\T_{zz}$ is often present as a result of cross-talk arising from the direct residual dipolar coupling of the qubits \cite{PhysRevApplied054023}, albeit non-controllable. 

\newpage

\section{Non-commutative Polynomial Optimisation}
\label{sec:nc-pop}

For the convenience of the reader, let us now present an outline of some basic definitions and results in non-commutative polynomial optimisation. We stress that this section is neither original material \cite{Pironio2010}, nor an exhaustive survey \cite{burgdorf2016optimization}, but rather a summary of some of the results of Pironio et al. \cite{Pironio2010}.

The non-commutative version of the polynomial optimization problem in Hilbert space $H$
with inner product 
$\langle\cdot,\cdot\rangle$, 
and a normalised vector $\phi$, i.e., $\|\phi \|^2=1$, is: 

\begin{equation}\notag
\mathbf{P}:\qquad
\begin{array}{llcll}
\quad p^\star&=&\displaystyle\min_{(H,X,\phi)}& \multicolumn{1}{l}{\langle \phi, p(X) \phi\rangle}\\
&&\text{s.t.}& q_i(X)\succeq
0&\quad i=1,\ldots,m\,,\\
\end{array}
\end{equation}

where $q_i(x)\succeq 0$ denotes that the operator $q_i(X)$ is positive semi-definite. 


\paragraph{Monomials}
We introduce the $\dagger$-algebra that can be viewed as conjugate transpose. Also, for each $X_i$ in $X$, there is a corresponding $X_i^{\dagger}$. For simplicity, 
let $[X,X^{\dagger}]$ denote those $2n$ operators.

A monomial $\omega(X)$ is defined as the product of powers of variables from $[X,X^{\dagger}]$. The empty monomial is 1. Since $X$ is non-commutative, two monomials with the same variables but a different order of variables are regarded as different monomials. Also, for monomials, we have $\omega^{\dagger}=\omega_r^{\dagger} \omega_{r-1}^{\dagger} \dots \omega_1^{\dagger}$ when $\omega=\omega_1 \omega_2 \dots \omega_r$.

The degree of a monomial, denoted by $|\omega|$, refers to the sum of the exponents of all operators in the monomial $\omega$. Let $\mathcal{W}_d$ denote the collection of all monomials whose degrees $|\omega|$ are less or equal to $d$.

A polynomial $p(X)$ of degree $d$ is defined to be a linear combination of monomials $\omega\in \mathcal{W}_d$ with the coefficients $p_{\omega}$, which lie in the field of real or complex numbers. Hence, $\mathcal{W}_d$ can also be understood as the monomial basis for polynomials of degree $d$.

Looking back to the problem \textbf{P}, if we assume that the degree of $p(X)$ and $q_i(X)$ are $\deg (p)$ and $\deg (q_i)$ respectively, then those non-commutative polynomials can be written as

\begin{equation}
    p(X)=\sum_{|\omega|\leq \deg(p)} p_{\omega} \omega,\quad
    q_i(X) = \sum_{|\mu|\leq \deg(q_i)} q_{i,\mu} \mu, 
    \label{LCoP}
\end{equation}

where $i = 1,\ldots,m$.

\paragraph{Moments}
With a feasible solution $(H,X,\phi)$ of problem \textbf{P},
we can define the moments on a field $\mathbb{C}$ as:

\begin{equation}
    y_{\omega} = \langle \phi, \omega(X) \phi \rangle, \label{DEFoMOMENT}
\end{equation}

for all $\omega \in \mathcal{W}_{\infty}$ and  $y_1=\langle \phi,\phi \rangle=1$.
Given a degree $r$, moments whose degrees are less or equal to $r$ form a sequence of $y=(y_{\omega})_{|\omega| \leq 2r}$.

\paragraph{Moment Matrices}

With a finite set of moments $y$ of degree $r$, we can define a corresponding order-$r$ moment matrix $M_r(y)$:

\begin{equation}
    M_r(y)(\nu,\omega) = y_{\nu^{\dagger}\omega} = \langle \phi, \nu^{\dagger}(X)\omega(X) \phi \rangle,
\end{equation}

for any $ |\nu|,|\omega| \leq r$, and a localising matrix $M_{r-d_i}(q_i y)$ whose entries are given by: 

\begin{equation}
    M_{r-d_i}(q_iy)(\nu,\omega) = \sum_{|\mu| \leq \deg(q_i)} q_{i,\mu} y_{\nu^{\dagger}\mu\omega} = \sum_{|\mu| \leq \deg(q_i)} q_{i,\mu} \langle \phi, \nu^{\dagger}(X) \mu(X) \omega(X) \phi \rangle,
\end{equation}

for any $|\nu|,|\omega| \leq r-d_i$, where $d_i=\lceil \deg(q_i)/2\rceil$. The upper bounds of $|\nu|$ and $|\omega|$ are lower than the that of the moment matrix because $y_{\nu^{\dagger}\mu \omega}$ is only defined on $\nu^{\dagger}\mu\omega \in \mathcal{W}_{2r}$ while $\mu\in \mathcal{W}_{\deg(q_i)}$.

\paragraph{Expressing the Objective}

We can obtain the so-called order-$r$ SDP relaxation of the non-commutative polynomial optimisation problem \textbf{P} by choosing an order $r$ that satisfies the condition of $2r\geq \max\{\deg(p),\deg(q_i)\}$. 
The objective $p(X)$ can be rewritten as:

\begin{align}
\label{eq:sdpobj}
\langle\phi,p(X)\phi\rangle = \langle\phi,\sum_{|\omega|\leq \deg(p)} p_{\omega} \omega (X) \phi\rangle = \sum_{|\omega|\leq \deg(p)} p_{\omega} \langle\phi, \omega (X) \phi\rangle = \sum_{|\omega|\leq \deg(p)} p_{\omega} y_{\omega}.
\end{align}

Next, we need to consider the positive-semidefinite (PSD) constraints.

\paragraph{PSD Moment Matrix}
By introducing $y=(y_{\omega})_{|\omega| \leq 2r}$ into the relaxation, we can relax: \begin{equation}
\label{Y_v*w}
    y_{\nu^{\dagger}\omega}=\langle \phi,\nu^{\dagger}(X)\omega(X)\phi\rangle, 
\end{equation}
for all $\nu^{\dagger}\omega \in \mathcal{W}_{2r}$ to a positive-semidefinite constraint on the moment matrix $M_r(y)$. Indeed, for any vector $z\in\mathbb{C}^{|\mathcal{W}_r|}$,  Pironio et al. \cite{Pironio2010} have shown:
\begin{eqnarray*}
z^{\dagger} M_r(y) z &=& \sum_{|\nu|\leq r} \sum_{|\omega|\leq r} 
z_{\nu}^{\dagger} M_r(y)(\nu,\omega) z_{\omega} = \sum_{|\nu|\leq r} \sum_{|\omega|\leq r} z_{\nu}^{\dagger} y_{\nu^{\dagger}\omega} z_{\omega} \\
&=& \sum_{|\nu|\leq r} \sum_{|\omega|\leq r} 
z_{\nu}^{\dagger} \langle \phi,\nu^{\dagger}(X)\omega(X)\phi \rangle z_{\omega} 
= \langle \phi,\left(\sum_{|\omega|\leq r} z_{\omega} \omega(X) \right)^{\dagger} \sum_{|\omega|\leq r} z_{\omega} \omega(X) \phi \rangle \geq 0
\end{eqnarray*}
\label{sec:localising}

\paragraph{PSD Localising Matrix}

The constraints $q_i(X)\succcurlyeq 0,\forall i=1,\ldots,m$ can also be relaxed to constraints such that the localise matrix $M_{r-d_i}(q_i y)$ is positive semidefinite, for $d_i=\lceil \deg(q_i)/2\rceil$.
Indeed, for any vector $z\in\mathbb{C}^{|\mathcal{W}_r|}$, Pironio et al.  \cite{Pironio2010} have shown:

\begin{eqnarray*}
z^{\dagger} M_{r-d_i}(q_i y) z &=& \sum_{|\nu|\leq r-d_i} \sum_{|\omega|\leq r-d_i} 
z_{\nu}^{\dagger} M_{r-q_i}(q_i y)(\nu,\omega) z_{\omega} 
= \sum_{|\nu|\leq r-d_i} \sum_{|\omega|\leq r-d_i} 
z_{\nu}^{\dagger} \left(
\sum_{|\mu| \leq \deg(q_i)} q_{i,\mu} y_{\nu^{\dagger}\mu\omega}  
\right) z_{\omega} \\
&=& \sum_{|\nu|\leq r-d_i} \sum_{|\omega|\leq r-d_i} 
z_{\nu}^{\dagger} \left(
\sum_{|\mu| \leq \deg(q_i)} q_{i,\mu} \langle \phi, \nu^{\dagger}(X) \mu(X) \omega(X) \phi \rangle
\right) z_{\omega} \\
&=& \langle \phi,\sum_{|\nu|\leq r-d_i} z_{\nu}^{\dagger}\nu^{\dagger}(X) 
\sum_{|\mu| \leq \deg(q_i)} q_{i,\mu} \mu(X)
\sum_{|\omega|\leq r-d_i} z_{\omega} \omega(X) \phi \rangle \\
&=& \langle \phi, \left(\sum_{|\omega|\leq r} z_{\omega} \omega(X)\right)^{\dagger}
q_i(X)
\sum_{|\omega|\leq r} z_{\omega} \omega(X) \phi \rangle \geq 0.
\end{eqnarray*}


\paragraph{Convergent SDP Relaxation}

By putting together the objective \eqref{eq:sdpobj} and the positive-semidefinite constraints on the moment and localising matrices, we can write down the order-$r$ SDP relaxation of the non-commutative polynomial optimisation problem as:

\begin{align}
	  \min_{y=(y_{\omega})_{|\omega|\leq 2r}} & 
	  \sum_{|\omega|\leq d} p_{\omega} y_{\omega} \\ 
	  M_r(X) & \succcurlyeq 0 \\
	  M_{r-d_i}(q_i X) & \succcurlyeq 0 \quad i=1,\ldots,m \\
	                       y_1 & = 1 \\
	  \langle\phi,\phi\rangle & = 1 \label{NCPO-R}.
\end{align}

If Archimedean Assumption \ref{Archimedean} is satisfied, Pironio et al.  \cite{Pironio2010} have shown that $\lim_{r \to \infty} p^r=p^*$.
Furthermore, if the so-called rank-loop condition is satisfied, i.e., when the rank of moment matrix $M_r(y)$ equals the rank of $M_{r-d}(y)$, where $d$ is the highest degree of constraints and $d\geq 1$, 
the global optimality is reached \cite{Pironio2010} at the order-$(r-d)$ relaxation. 

\newpage

\section{Computability and Convergence Rates}
\label{sec:computability}

In some sense, our paper could be seen a proof of approximability of quantum optimal control to any precision in a model, where real numbers are atomic units of information and basic floating-point arithmetic operations on real numbers can be performed in unit time.
In particular, from Theorem \ref{T2} and Corollary \ref{T1}, it follows that arbitrarily good approximation of quantum optimal control is computable in the  model of Blum, Shub, and Smale \cite{blum1989}:

\begin{proposition}
For any number of terms $k$,
any order $r$, 
any step $\Delta t$ of the discretisation of time,
and any finite precision $\epsilon$,
the semidefinite programming relaxation of either Theorem \ref{T2} or Corollary \ref{T1} 
is computable up to the precision $\epsilon$ in time polynomial in its dimension 
in the Blum-Shub-Smale model.
\end{proposition}

\begin{proof}
By iteration analysis of primal-dual interior-point methods, cf. \cite{Tuncel2000}, 
applied to the SDP \eqref{NCPO-R}.
For an explicit treatment, using a substantially more complicated arguments, see \cite{muramatsu2018oracle}.
We only have to notice that the number of localising constraints is polynomial in the dimension. 
\end{proof}

To prove similar results in the Turing model is a major challenge. Indeed, quantum optimal with a control signal taking on values from a discrete set is undecidable in the Turing model  \cite{bondar2019uncomputability}
without additional assumptions.



Beyond the computability, further study of the run-time and the rate of convergence would be beneficial. 
To reduce run-time, it seems important
 to exploit sparsity \cite{mevissen2008solving,klep2019sparse,wang2019tssos}. 
The rate of convergence depends on the rate of convergence of Magnus expansion, 
the choice of the discretisation of time \cite{singh2018high,kopylov2019magnus}, 
the rate of convergence of the hierarchy of SDP relaxations to the optimum of the non-commutative polynomial optimisation problem,
and possibly the rate of convergence of the SDP solver. 
Considering the growth of the dimensions of the SDP relaxations, 
the rate of convergence is 
both of practical and theoretical interest.
The rates of convergence of Magnus expansion in $m$, depending on the discretisation, are well known \cite[Chapter 9]{singh2018high}.
For another hierarchy of SDP relaxations, under some additional assumptions about boundedness, the order-$r$ semidefinite programming relaxation of the comutative polynomial optimisation problem has error bounded by $O(1/r^2)$ in the worst case. See  \cite{deKlerk2020,slot2019improved,laurent2020nearoptimal}. 
In the non-commutative case of \eqref{polyncprog2},
the present-best result is that of Fang and Fawzi \cite{fang2019sum}, who again under some additional assumptions,
show that the rate of convergence is no worse than $O(1/r^2)$
for the order $r$.
Continuing this line of work, one could estimate the overall convergence rate,
which would be a major result.

}

\bibliography{ref}

\end{document}